\documentclass[12pt,reqno]{amsart}
\allowdisplaybreaks[4]

\usepackage [latin1]{inputenc}
\usepackage{graphicx}
\usepackage{amssymb}
\usepackage{color}
\usepackage{epstopdf}
\usepackage{amsthm,amsmath,amssymb}
\usepackage{mathrsfs}
\usepackage{cite}
\usepackage{subeqnarray}
\usepackage{cases}

\newtheorem{proposition}{Proposition}[section]
\newtheorem{definition}{Definition}[section]

\newtheorem{lemma}{Lemma}[section]
\newtheorem{remark}{Remark}[section]
\newtheorem{theorem}{Theorem}[section]

\begin{document}

\begin{center}
{\large \sc \bf The integrable hierarchy and the nonlinear Riemann-Hilbert problem associated with one typical Einstein-Weyl physico-geometric dispersionless system}

\vskip 20pt

{\large   }

\vskip 20pt

{Ge Yi, Tangna Lv, Kelei Tian* and Ying Xu \\

\it
 Hefei University of Technology, Hefei, Anhui 230601, China}

\bigskip

\bigskip
$^*$ Corresponding author:  {\tt kltian@ustc.edu.cn, kltian@hfut.edu.cn}
\bigskip

\bigskip

{\today}

\end{center}

\bigskip
\bigskip
\textbf{Abstract:} From a specific series of exchange conditions for a one-parameter Hamiltonian vector field, we establish an integrable hierarchy using Lax pairs derived from the dispersionless partial differential equation. An exterior differential form of the integrable hierarchy is introduced, further confirming the existence of the tau function. Subsequently, we present the twistor structure of  the hierarchy. By constructing the nonlinear Riemann Hilbert problem for the equation, the structure of the solution to the equation is better understood.

\bigskip

\textit{\textbf{Keywords:}} integrable hierarchy, Hamiltonian vector field, tau function, twistor structure, nonlinear Riemann Hilbert problem.
\bigskip
\bigskip

\section{\sc \bf Introduction}

Einstein-Weyl geometry was first introduced by Weyl in the early 20th century and has been extensively studied \cite{WardRS, Hitchin}. The relations of dispersionless integrable systems to Einstein-Weyl geometry have been discussed in \cite{CaldPed, Dunajski1, Dunajski2, DunMason}, and there are several classes of PDE whose integrability can be seen in their formal linearized geometries \cite{FerKru}. One class of integrable models possessing the central quadric ansatz, we classifyied on the method of hydrodynamic reductions \cite{FerHuard}, leading to five regular forms. Then, applying central quadric transformation can obtain all Painlev\'e equations $P_{I}-P_{VI}$. These five regular forms appear in different forms in the classification of multidimensional integrable systems, including the well-known BF (Boyer-Finley) equation 
\begin{eqnarray}
    \left ( e^{u} \right )_{tt}+u_{xx}+u_{yy}=0, \nonumber
\end{eqnarray}
and dKP (dispersionless Kadomtsev-Petviashvili) equation
\begin{eqnarray}
    u_{xt}-\left ( uu_{x} \right )_{x} -u_{yy}=0.\nonumber
\end{eqnarray}
In this paper, we will focus on the equation related to physico-geometric as follows 
\begin{eqnarray}
    u_{xt}+\frac{1}{2} u_{yy}+\frac{1}{2} {\left (u ^{2}  \right ) } _{xy} =0.\label{1}
\end{eqnarray}
It is the representative case in the classification of infinite hydrodynamic chains satisfying the Egorov property. The concept of Egorov hydrodynamic chains has been introduced in \cite{Pavlov1, Pavlov2}, and shows their connection to the integrable (2+1)-dimensional equations of hydrodynamic type. The results obtained in the classification problem for integrable (2+1)-dimensional systems of hydrodynamic type are described. In fact, the conformal structures corresponding to this integrable equation satisfy the Einstein-Weyl property. This construction uses the gauge group SDiff $(\Sigma ^{2})$ modelled on Riccati spaces and provides an Einstein-Weyl structure from the solutions of gauge field equations \cite{Cald}. In this work, Calderbank points out that, based on generalized Nahm equations with the gauge group SDiff $(\Sigma ^{2})$, we can obtain generalized Nahm equations corresponding to equation \eqref{1} by twistor theoretical approach.

The study of the integrable systems starts from the 19th century in the field of mathematical physics, and it is essential investigation of nonlinear evolution equations. These integrable systems have special nonlinear evolution equations in which the shape and behaviour of the waves remain constant throughout their evolution \cite{PDLax, EAZabo}. In the 1970s, it was observed that some nonlinear equations possessed more symmetries than the traditional Lie group symmetries \cite{Olver}. Later, Mikhailov, Shabat and Yamilov used a bilinear approach based on the Lax pairs, after which an infinite number of conservation laws and symmetries could be extracted from the nonlinear problem by introducing suitable variables \cite{MiShYa}. In the 1990s, the Lax equation and Hamiltonian structures of dispersionless integrable systems attracted increasing attention \cite{Kri}. These works provide a foundation for the study of the properties of dispersionless equations, especially their infinite symmetries and the corresponding hierarchies. Actually, the integrable hierarchies with Hamiltonian vector field forms are more abundantly studied, for instance, the dKP hierarchies, SDiff(2) Toda hierarchies \cite{Take1, Take2, Take3}, and the dDS(dispersionless Davey-Stewartson) hierarchies \cite{GeYi}. 

Sato theory is widely recognized as one of the key theories in the field of integrable systems. Specifically, the tau function theory, which is the core of Sato theory, is essential to many related subjects and facilitates the analysis of PDEs. As an effective tool for studying nonlinear partial differential equations, the tau function theory covers a variety of topics including soliton theory, inverse scattering methods, Hamiltonian methods, etc \cite{DaKaM1, DaKaM2, FadTak}. In addition to its applications in integrable systems, the tau function theory has significant connections to string theory. Furthermore, the tau function theory, which has been extensively utilized to explain the vibrational modes and interactions of strings in physics, provides an important link between integrable systems and string theory\cite{DeSa, Danie}.

Additionally, several aspects of the study of dispersionless integrable systems are drawing more attention including the long-time behaviour and the possible wave breaking properties of dispersionless equations. Recently, Manakov and Santini introduced a new IST(inverse scattering transform) method for solving the dispersionless integrable equations at the formal level, such as the dKP equation, the heavenly equation of Plebanski, the Pavlov equation and so on \cite{ManSan, MaSa, Manakov, kovtini, YiSan, akovtini}. One key step in the Manakov-Santini method is to establish the relevant nonlinear Riemann-Hilbert problem with the dispersionless equation.  Two important dispersionless integrable systems, the (3+1)-dimensional Dunajski hierarchy and the (2+1)-dimensional dDS system, are discussed by using the Manakov-Santini method \cite{YiSan, GYi}. In \cite{Take3}, an alternate construction of the twistor structure, which provides general solutions for dispersionless systems, is provided by Takasaki and Takebe. This is a kind of Riemann-Hilbert construction. Originally, the Riemann-Hilbert problem was a series of equations pertaining to the properties of analytic functions. With the extensive study of integrable dynamical systems, it gradually associates with the integrable systems. Actually, in the study of classical integrable systems, we usually focus on the linearised Riemann-Hilbert problem, but in the exploration of dispersionless integrable systems, the relevant nonlinear Riemann-Hilbert problem draws more attention.

The paper is organized as follows. In Section 2, we construct the related hierarchy, which contains an infinite number of compatible partial differential equations. In Section 3, we show the existence of the tau function. In Section 4, we discuss the twistor structure of the integrable hierarchy. In Section5, we consider the relevant nonlinear Riemann-Hilbert problem.

\bigskip

\section{\sc \bf  Lax Formalism and hierarchy}
In this section, we construct a hierarchy associated with the dispersionless system by taking a truncation of the eigenfunction from the Lax pair of equation \eqref{1}. This hierarchy contains an infinite number of integrable dispersionless equations, which are infinite symmetries of equation \eqref{1}. In \cite{FerKru}, the authors propose several classes of integrable models, and the integrability of these dispersionless equations is reflected in the compatibility of the vector field Lax pair. In particular, the Hamiltonian vector field Lax pair of equation \eqref{1} reads as follows
\begin{eqnarray}
     L_{1} &=&\partial _{t} -\left \{ H_{1}, \cdot \right \}=\partial _{y} -\lambda \partial _{x} +2u_{x} \lambda \partial _{\lambda}, \label{L1}\\
     L_{2} &=&\partial _{y} -\left \{ H_{2}, \cdot \right \}=\partial _{t} +\left ( \frac{1}{2} \lambda ^{2} +u \lambda \right ) \partial _{x} -\left ( u_{x} \lambda +u_{y} +2uu_{x}   \right ) \lambda \partial _{\lambda }, \label{L2}
\end{eqnarray}
in which Hamiltonians $H_{1}$ and $H_{2}$ are given by
\begin{eqnarray}
    H_{1}&=&\lambda +2u\label{hami1},\\
    H_{2}&=&-\left ( \frac{1}{4}\lambda ^{2} +u\lambda +u^{2} + \partial _{x}^{-1}u_{y} \right )\label{hami2}.
\end{eqnarray}
Here and here after, in this paper, the symbol $\left \{ , \right \}$ stands for the Poisson bracket in $2D$ phase space ${\left ( \lambda ,x \right )}$ as follows
\begin{eqnarray}
    \left \{ A,B \right \} =\lambda \frac{\partial A}{\partial \lambda } \frac{\partial B}{\partial x} -\lambda \frac{\partial A}{\partial x }\frac{\partial B}{\partial \lambda }.\nonumber
\end{eqnarray}
As we know, the dispersionless system \eqref{1} arises from the commutation condition
\begin{eqnarray}
    \left [ L_{1} ,L_{2} \right ] =0.\label{l12}
\end{eqnarray}
In fact, the commutation condition \eqref{l12} is equivalent to the following Zakharov-Shabat equation
\begin{eqnarray}
    \frac{H_{1} }{\partial y}-\frac{H_{2} }{\partial t}+\left \{ H_{1},H_{2} \right \}=0\label{H12}.
\end{eqnarray}

In order to constructe infinite symmetries for dispersionless system \eqref{1}, we first consider an eigenfunction of the vector field. With a given simpel closed curve around the original point in the complex $\lambda$ -plane, the eigenfunction $\Phi$ of $L_{1}$ can be described by the following formal Laurent expansion
\begin{eqnarray}\label{phi}
    \Phi=\lambda +\sum_{k\le 0}a_{k}\lambda ^{k}.\label{phi1}
\end{eqnarray}
In fact, its coefficients $a_{k}$ can be calculated directly from $L_{1}\Phi =0$, in which all coefficients rely on the function $u$ and their derivatives or integrals with respect to the independent variables under $x$, $y$. For example,  
\begin{eqnarray}
    a_{0}&=&2u,\nonumber\\
    a_{-1}&=&2\partial _{x}^{-1}u_{y},\nonumber \\
    a_{-2}&=&2\partial _{x}^{-2}u_{yy}+4\partial _{x}^{-1}\left ( u_{x} \partial _{x}^{-1}u_{y} \right ).\nonumber
\end{eqnarray}
Then we can recursively work out arbitrary coefficients of the eigenfunction $\Phi$.
\begin{remark}
In constructing the hierarchy of the dispersionless system \eqref{1}, we need to consider $m$ powers of the eigenfunction $\Phi$. Based on the coefficients $a_{k}$ of the eigenfunction $\Phi$ obtained above, denote $\Phi ^{m}$ as follows
\begin{eqnarray}
   \Phi ^{m}=\lambda ^{m}+\sum_{k\le m-1} q_{k}^{\left ( m \right ) }\lambda ^{k},\;\;\;m=1,2,\cdots,\nonumber
\end{eqnarray} 
in which $q_{k}^{(m)}=q_{k}^{(m)}(a_{0},a_{-1},\dots,a_{k-1})$.

Meanwhile, according to the expression \eqref{phi1}, conversely, the parameter $\lambda$ can be formally represented by $\Phi$ as the following Laurent expansion
\begin{eqnarray}
   \lambda =\Phi + \sum_{k\le 0}p_{k}\Phi ^{k},\label{lambda} 
\end{eqnarray}
where $p_{k}=p_{k}\left ( a_{0},a_{-1},\dots,a_{k} \right )$. 
\end{remark}
Next, we start to construct infinite symmetries of dispersionless system \eqref{1} which form the integrable hierarchy, and the Lax formalisms are vector fields which all have Hamiltonian structures. More specifically these Hamiltonians are truncations of higher powers of the eigenfunction $\Phi$. Furthermore, we prove that all vector fields share the same eigenfunctions.
\begin{definition}
   For any positive integer $m$, we define a series of the Hamiltonian vector fields as follow
\begin{eqnarray}
   L_{m}=\partial _{t_{m} }-\left \{ H_{m}, \cdot \right \},\label{Lm0}
\end{eqnarray} 
where
\begin{eqnarray}
    H_{m}=(-1)^{m-1}\left ( \frac{1}{m^{2} } \Phi^{m} \right )_{\ge 0},\label{Hm0}
\end{eqnarray}
and $\left ( \right )_{\ge 0}$ denotes the projection of the polynomial in $\lambda$ with a positive power. In particularly, take $t_{1}=t$, $t_{2}=y$.
\end{definition}

Now, we discuss the compatibility of the operator $L_{1}$ with this series of operators, with the following conclusion.
\begin{proposition}
    The Zakharov-Shabat equations
\begin{eqnarray}
    \frac{\partial H_{1}}{\partial t_{m} }-\frac{\partial H_{m}}{\partial y }+\left \{ H_{1},H_{m} \right \}=0 \label{H1m}
\end{eqnarray}
are equivalent to the following flow equations
\begin{eqnarray}
    2u_{t_{m}}-\left ( -1 \right )^{m-1}q_{0,y}^{(m)}=0.\label{liu}
\end{eqnarray}
\end{proposition}
\begin{proof}
   In fact, the Zakharov-Shabat equations \eqref{H1m} represent 
\begin{eqnarray}\label{L1Hm}
    \frac{\partial H_{1} }{\partial t_{m} }&=&L_{1}\left ( H_{m}  \right )\nonumber\\
    &=&L_{1}\left [ \left ( -1 \right )^{m-1}\left ( \frac{1}{m^{2}} \Phi ^{m} \right )_{\ge 0} \right ]\nonumber\\
    &=&\lambda +(-1)^{m-1}2u_{t_{m} }.
\end{eqnarray}
Since the eigenfunction space is a ring, then $\left ( -1 \right )^{m-1}\frac{1}{m^{2}} \Phi ^{m}$ are also eigenfunctions and satisfy the equation $L_{1}(\left ( -1 \right )^{m-1}\frac{1}{m^{2}} \Phi ^{m})=0$. The following recursive relations are obtained as
\begin{eqnarray}
    \begin{aligned}\label{qlm}
       &q_{l,y}^{(m)}-q_{l-1,x}^{(m)}+2u_{x}l q_{l}^{(m)}=0,\\
       &q_{m-1}^{(m)}=\frac{2u}{m},\;\;\;\;q_{l}^{(m)}=0 , \; \left ( l>m-1  \right ).  
    \end{aligned}   
\end{eqnarray}
Bringing the above recursive relations into the display formula \eqref{L1Hm}, then gives the flow equations \eqref{liu}. 
\end{proof}
Obviously, Zakharov-Shabat equations \eqref{H1m} are equivalent to the following commutation condition
\begin{eqnarray}
    \left [ L_{1} ,L_{m} \right ] =0.\label{Lm1}
\end{eqnarray}
As we know, the space of vector field eigenfunctions is a ring whose base consists of two independent eigenfunctions. Combining the above commutation \eqref{Lm1}, for any positive integers $n$, the vector fields $L_{m}$, $L_{n}$ share the eigenfunction space with $L_{1}$. This implies that $L_{n}$ satisfy the commutation condition with $L_{m}$. Thus we can define the hierarchy of dispersionless system \eqref{1}.
\begin{theorem}
   Zakharov-Shabat equations
\begin{eqnarray}
   \frac{\partial H_{m}}{\partial t_{n} }-\frac{\partial H_{n}}{\partial t_{m} }+\left \{ H_{m},H_{n} \right \}=0,\;\;\;m,n=1,2,\cdots, \label{Hmn} 
\end{eqnarray}
are equivalent to a hierarchy of compatible systems.
\end{theorem}
\begin{proof}
In fact, Zakharov-Shabat equations \eqref{Hmn} are equivalent to the commutation condition as follows
\begin{eqnarray}
    \left [ L_{m},L_{n} \right]=0 .
\end{eqnarray}
The proof is given below. For Zakharov-Shabat equations
\begin{eqnarray}
   &&\frac{\partial H_{m}}{\partial t_{n} }-\frac{\partial H_{n}}{\partial t_{m} }+\left \{ H_{m},H_{n} \right \}\nonumber\\
   &&=\left ( \frac{1}{m}\Phi^{m-1}\Phi_{t_{n}} \right )_{\ge 0}-\left ( \frac{1}{n}\Phi^{n-1}\Phi_{t_{m}} \right )_{\ge 0}+\left \{ H_{m},H_{n} \right \}\nonumber\\
   &&=\left [ \frac{1}{m}\Phi ^{m-1}\left \{ H_{n},\Phi \right \}-\frac{1}{n}\Phi ^{n-1}\left \{ H_{m},\Phi \right \}+\left \{ H_{m},H_{n} \right \} \right ]_{\ge 0}\nonumber\\
   &&=\left [ \left \{ H_{n},\frac{1}{m^{2}}\Phi ^{m} \right \}-\left \{ H_{m},\frac{1}{n^{2}}\Phi ^{n} \right \}+\left \{ H_{m},H_{n} \right \} \right ]_{\ge 0}\nonumber\\
   &&=\left [ \left \{ H_{n}-\frac{1}{n^{2}}\Phi ^{n},\frac{1}{m^{2}}\Phi ^{m} \right \}-\left \{ H_{m},H_{n}- \frac{1}{n^{2}}\Phi ^{n} \right \} \right ]_{\ge 0}\nonumber\\
   &&=\left [ \left \{ \left ( H_{n} \right )_{<  0},\frac{1}{m^{2}}\Phi ^{m} \right \}+\left \{ H_{m},\left ( H_{n} \right )_{<  0} \right \} \right ]_{\ge 0}\nonumber\\
   &&=\left [ \left \{ \left ( H_{m} \right )_{< 0},\left ( H_{n} \right )_{< 0} \right \} \right ]_{\ge 0},\nonumber
\end{eqnarray}
whose value equals to 0.

On the other hand,
\begin{eqnarray}
    \left [ L_{m},L_{n}\right ]&=&L_{m}L_{n}-L_{n}L_{m}\nonumber\\
    &=&\left ( \partial_{t_{m}}-\left \{ H_{m},\cdot \right \} \right )\left (\partial_{t_{n}}-\left \{ H_{n},\cdot \right \} \right )-\left (\partial_{t_{n}}-\left \{ H_{n},\cdot \right \} \right )\left ( \partial_{t_{m}}-\left \{ H_{m},\cdot \right \} \right )\nonumber\\
    &=&\left ( \partial_{t_{m}}-\lambda H_{m,\lambda }\partial _{x}+\lambda H_{m,x}\partial _{\lambda }\right )\left ( \partial_{t_{n}}-\lambda H_{n,\lambda }\partial _{x}+\lambda H_{n,x}\partial _{\lambda }\right )\nonumber\\
    &&-\left ( \partial_{t_{n}}-\lambda H_{n,\lambda }\partial _{x}+\lambda H_{n,x}\partial _{\lambda }\right )\left ( \partial_{t_{m}}-\lambda H_{m,\lambda }\partial _{x}+\lambda H_{m,x}\partial _{\lambda }\right )\nonumber\\
    &=&\lambda \left ( \frac{\partial H_{m}}{\partial t_{n} }-\frac{\partial H_{n}}{\partial t_{m} }+\left \{ H_{m},H_{n} \right \} \right )_{\lambda }+\lambda \left ( \frac{\partial H_{n}}{\partial t_{m} }-\frac{\partial H_{m}}{\partial t_{n} }+\left \{ H_{n},H_{m} \right \} \right )_{x}.\nonumber
\end{eqnarray}
Likewise, this value equals to 0.
\end{proof}

\begin{remark}
For the flow equations \eqref{liu}, when $m=1,2,3,$ bring them in the recursive relations \eqref{qlm},  
\begin{eqnarray}
    &&m=1,\;\;\;u_{t_{1}}=u_{y},\\
    &&m=2,\;\;\;2u_{xt_{2}}+u_{yy}+(u^{2})_{xy}=0,\label{L12}\\
    &&m=3,\;\;\;3u_{xt_{3}}-2(u^{2})_{yy}-2(u_{x}\partial_{x}^{-1}u_{y})_{y}-4(u^{2}u_{x})_{y}-\partial _{x}^{-1}u_{yyy}=0.\label{L13}
    \end{eqnarray}
Equation \eqref{L12} is the commutation condition $\left [ L_{1},L_{2} \right ]=0$ arising from Hamiltonian vector field Lax pair. Additionally, equation \eqref{L13} is equivalent to $\left [ L_{1},L_{3} \right ]=0$.
\end{remark}
To give an example, by taking $m=3$, then the Hamiltonian \eqref{Hm0} can be  written as
\begin{eqnarray}
    H_{3}=\frac{1}{9}\lambda ^{3}+\frac{1}{3}a_{0}\lambda ^{2}+\frac{1}{3}a_{0}^{2}\lambda +\frac{1}{3}a_{1}\lambda +\frac{2}{3}a_{0}a_{1}+\frac{1}{3}a_{2}+\frac{1}{9}a_{0}^{3}.\nonumber
\end{eqnarray}
The following nonlinear system 
\begin{eqnarray}
    u_{yt_{3} } +\left ( u^{2} \right )_{xt_{3} }+2\left ( u^{2} \right )_{yt_{2} }+4\left ( u^{2}u_{x} \right )_{t_{2} }+\partial _{x}^{-1}u_{yyt_{2} }+2\left ( u_{x}\partial _{x}^{-1}u_{y} \right )_{t_{2} }=0 \nonumber
\end{eqnarray}
derives from the Zakharov-Shabat equations \eqref{Hmn}, which equals to the commutation condition $\left [ L_{2},L_{3} \right ]=0$.

\bigskip

\section{\sc \bf  The tau function of the hierarchy}
In this section, we present two exterior differential 2-form which are equivalent to the hierarchy. In addition to the eigenfunction $\Phi$, we will construct another independent eigenfunction $\Psi$. Based on this, we obtain the $S$ function and the existence of the significant tau function.
\begin{definition}
   Introducing an exterior differential 2-form
\begin{eqnarray}
    \omega = \sum_{n=1}^{\infty }dH_{n}\wedge dt_{n}=\frac{d\lambda }{\lambda }\wedge dx+\sum_{n=2}^{\infty }dH_{n}\wedge dt_{n},\label{w1}
\end{eqnarray}
where ``d" denotes the full differentiation. 
\end{definition}
In fact, this exterior differential 2-form $\omega$ is obviously a closed form, hence,
\begin{eqnarray}
    d\omega =0.\nonumber
\end{eqnarray}
By Zakharov-Shabat equations \eqref{Hmn}, $\omega$ satisfies the following relation
\begin{eqnarray}
    \omega\wedge \omega  =0.\nonumber
\end{eqnarray}

The above two relations indicate the existence of two functions $P$ and $Q$ which give a pair of Darboux coordinates as
\begin{eqnarray}
    \omega =\frac{dP}{P} \wedge dQ.\nonumber
\end{eqnarray}

\begin{proposition}
    Taking $P=\Phi$, there exists another eigenfunction $\Psi$ satisfying
    \begin{eqnarray}
        \omega =\frac{d\Phi }{\Phi } \wedge d\Psi,\label{w2}
    \end{eqnarray}
    in which 
    \begin{eqnarray}
       \Phi&=&\lambda +\sum_{k\ge 0}a_{k}\lambda ^{-k}, \label{Phi}\\ 
       \Psi&=&x+\sum_{n=1}^{\infty }\frac{t_{n} }{n}\Phi ^{n}+\sum_{i=1}^{\infty }v_{i}\Phi ^{-i}.\label{Psi}
    \end{eqnarray}
\end{proposition}
In particularly, 
\begin{eqnarray}
    \left \{ \Phi ,\Psi  \right \}=\Phi,\nonumber
\end{eqnarray}
and they are the eigenfunctions of the vector field $L_{m}$, satisfying $L_{m}\left ( K \right )=0$, in which $K=\Phi, \Psi$. Then the Lax expressions read as
\begin{eqnarray}
    \frac{\partial \Phi }{\partial {t_{n} } } =\left \{ H_{n},\Phi \right \},\;\;\;\frac{\partial \Psi }{\partial {t_{n} } } =\left \{ H_{n},\Psi  \right \}. \label{hierarchy}
\end{eqnarray}
are equivalent to the hierarchy of dispersionless system \eqref{1}.

Next, we prove the equivalence of two forms of 2-form $\omega$. Firstly, we show that the known Lax hierarchy derives the exterior differential equations. Looking back at the equation \eqref{w1} and equation \eqref{w2}, they can both be written as a linear combination of $d\lambda \wedge dx$, $d\lambda \wedge dt_{n}$, $dx \wedge dt_{n}$ and $dt_{m}\wedge dt_{n}$. Obviously, there are coefficients of $d\lambda \wedge dx$ for 
\begin{eqnarray}
    \frac{1}{\Phi }\begin{vmatrix}
            \Phi_{\lambda }& \Phi_{x}   \\
            \Psi_{\lambda }&\Psi _{x} 
            \end{vmatrix}&=&\frac{1}{\lambda }.\nonumber
\end{eqnarray}
Similarly, we have the following coefficients of $d\lambda \wedge dt_{n}$, $dx \wedge dt_{n}$ and $dt_{m}\wedge dt_{n}$ as
\begin{eqnarray}
    \frac{1}{\Phi } \begin{vmatrix}
            \Phi _{\lambda} & \Phi_{t_{n}} \\
            \Psi _{\lambda} & \Psi_{t_{n}}\nonumber\\  
            \end{vmatrix}&=&\frac{1}{\Phi}\left (\Phi_{\lambda}\left \{ H_{n},\Psi \right \}-\left \{ H_{n},\Phi\right \}\Psi_{\lambda}\right )\nonumber\\
            &=&\frac{1}{\Phi}\left (H_{n,\lambda }\left \{ \Phi ,\Psi \right \}\right )=\frac{\partial H_{n} }{\partial \lambda },\nonumber\\
     \frac{1}{\Phi } \begin{vmatrix}
            \Phi _{x} & \Phi_{t_{n}} \\
            \Psi _{x} & \Psi_{t_{n}}\nonumber\\  
            \end{vmatrix}&=&\frac{1}{\Phi}\left (\Phi_{x}\left \{ H_{n},\Psi \right \}-\left \{ H_{n},\Phi\right \}\Psi_{x}\right )\nonumber\\
            &=&\frac{1}{\Phi}\left (H_{n,x }\left \{ \Phi ,\Psi \right \}\right )=\frac{\partial H_{n} }{\partial x },\nonumber\\
    \frac{1}{\Phi } \begin{vmatrix}
            \Phi _{t_{m}} & \Phi_{t_{n}} \\
            \Psi _{t_{m}} & \Psi_{t_{n}}\nonumber\\  
            \end{vmatrix}&=&\frac{1}{\Phi}\left (\left \{H_{m},\Phi \right \} \left \{ H_{n},\Psi \right \}-\left \{H_{n},\Phi \right \} \left \{ H_{m},\Psi \right \}\right )\nonumber\\
            &=&\frac{1}{\Phi}\left (\left \{ \Phi ,\Psi  \right \}\left \{ H_{m},H_{n}\right \} \right )=\left \{ H_{m},H_{n}\right \}.\nonumber
\end{eqnarray}
Secondly, the Lax hierarchy is deduced in turn from the exterior differential equations. The coefficients of $d\lambda \wedge dx$ hold as
\begin{eqnarray}
    \Phi =\lambda \left (\Phi _{\lambda }\Psi _{x}-\Psi _{\lambda }\Phi _{x} \right )=\left \{ \Phi,\Psi  \right \}. \nonumber
\end{eqnarray}
Comparing the coefficients on both sides of $d\lambda \wedge dt_{n}$, $dx \wedge dt_{n}$, there are
\begin{eqnarray}
    \frac{1}{\Phi } \left (\Phi _{\lambda }\Psi _{t_{n} }-\Psi _{\lambda }\Phi _{t_{n} } \right )=H_{n,\lambda }, \nonumber\\
    \frac{1}{\Phi } \left (\Phi _{t_{n} }\Psi _{x}-\Psi _{t_{n} }\Phi _{x} \right )=H_{n,x}. \nonumber
\end{eqnarray}
For the above equations, we can accurately figure out $\Phi _{t_{n}}$ and $ \Psi _{t_{n}}$, they are given by
\begin{eqnarray}
    \Phi _{t_{n}}=\lambda \left ( H_{n,\lambda } \Phi _{x}-H_{n,x } \Phi _{\lambda } \right )=\left \{ H_{n},\Phi  \right \}, \nonumber\\
    \Psi _{t_{n}}=\lambda \left ( H_{n,\lambda } \Psi _{x}-H_{n,x } \Psi _{\lambda } \right )=\left \{ H_{n},\Psi  \right \}. \nonumber
\end{eqnarray}
The process of proving the equivalence of exterior differential equations is complete.

Combining the two equivalent definitions of the 2-form $\omega$, the following equation holds
\begin{eqnarray}
    d\Psi \wedge \frac{d\Phi }{\Phi }+\frac{d\lambda }{\lambda }\wedge dx+\sum_{n=2}^{\infty }dH_{n}\wedge dt_{n}=0.\nonumber
\end{eqnarray}
This equation also can be written as 
\begin{eqnarray}
    d\left ( \Psi d\log{\Phi } +\log{\lambda } dx+\sum_{n=2}^{\infty }H_{n}dt_{n} \right ) =0.\nonumber
\end{eqnarray}
This implies the existence of a function $S$ such that
\begin{eqnarray}
    dS=\Psi d\log{\Phi } +\log{\lambda } dx+\sum_{n=2}^{\infty }H_{n}dt_{n}.\nonumber
\end{eqnarray}
Then  there are the following equations
\begin{eqnarray}
    \frac{\partial S}{\partial x}=\log{\lambda },\;\;\;\frac{\partial S}{\partial t_{n} }=H_{n},\;\;\;\frac{\partial S}{\partial \log{\Phi } }=\Psi.\label{2}
\end{eqnarray}
In fact, even though $S$ is not a true potential function, it has some profound implications in the hierarchy. Then $S$ can be expressed in the form of the Laurent series about $\Phi$, $\Psi$.

\begin{proposition}
    $S$ is given by
    \begin{eqnarray}
        S=x\log{\Phi }+\sum_{n=1}^{\infty }\frac{t_{n} }{n^{2} }\Phi ^{n}+\sum_{i=1}^{\infty }S_{i}\Phi ^{-i},\;\;\;S_{i}=-\frac{v_{i} }{i}. \label{S}
    \end{eqnarray}
    Through the equations \eqref{2} and \eqref{S} show that $H_{n}$ can be written in the form of a Laurent series related to $S_{i}$ as follows
    \begin{eqnarray}
        H_{n}=\frac{\partial S}{\partial t_{n} }=\frac{\Phi ^{n} }{n^{2} }+\sum_{i=1}^{\infty }\frac{\partial S_{i} }{\partial t_{n} }\Phi ^{-i}.\label{Hn}
    \end{eqnarray}
    In particular, for the above expression of $H_{n}$, when $n=1$, we have
    \begin{eqnarray}
        H_{1}=\lambda =\Phi +\sum_{i=1}^{\infty }\frac{\partial S_{i} }{\partial x}\Phi ^{-i},\nonumber
    \end{eqnarray}
    which is also equivalent to equation \eqref{lambda}.
\end{proposition}

Next, introduce the 1-form residue operator
\begin{eqnarray}
    res\sum a_{n}\lambda ^{n}=a_{-1}.\nonumber
\end{eqnarray}
It has the following properties.
\begin{lemma}\label{lemma1}
    For any Laurent series $L$ and $M$ of $\lambda$,
    \begin{eqnarray}
        &&res d_{\lambda }L=0,\nonumber\\
        &&res Ld_{\lambda }M=-res Md_{\lambda }L,\nonumber\\
        &&res Ld_{\lambda }M=res \left ( L_{\ge 0} \right )d_{\lambda }\left ( M_{\le -1} \right )+res \left ( L_{\le-1} \right )d_{\lambda }\left ( M_{\ge  0} \right ).\nonumber 
    \end{eqnarray}
\end{lemma}
\begin{lemma}
    For any positive integer $n$
    \begin{eqnarray}
        res \Phi ^{n-1}d_{\lambda }\Phi=\delta _{n,-1}.\nonumber
    \end{eqnarray}
\end{lemma}
From the above properties of the residue operator, we can obtain the following lemma and prove that.
\begin{lemma}
    For the Laurent series $\Psi$, the partial derivatives of the coefficients $v_{i}$ 
    \begin{eqnarray}
        \frac{\partial v_{i} }{\partial t_{n} }=res \Phi ^{i}d_{\lambda }H_{n}.\nonumber
    \end{eqnarray}
\end{lemma}   
\begin{proof}
    For the Laurent series $\Psi$, from the chain rule,
    \begin{eqnarray}
        \frac{\partial \Psi}{\partial t_{n} }=\frac{1}{n}\Phi ^{n}+\frac{\partial \Psi }{\partial \Phi }\frac{\partial \Phi }{\partial t_{n} }+\sum_{i=1}^{\infty }\frac{\partial v_{i} }{\partial t_{n} }\Phi ^{-i},\nonumber
    \end{eqnarray}
    in which 
    \begin{eqnarray}
        \frac{\partial \Psi }{\partial \Phi }=\sum_{n=1}^{\infty }t_{n}\Phi ^{n-1}-\sum_{i=1}^{\infty }iv_{i}\Phi ^{-i-1}.\nonumber
    \end{eqnarray}
    Then 
    \begin{eqnarray}
        \frac{\partial v_{i}}{\partial t_{n}}&=&res \Phi ^{i-1} \left (\frac{\partial \Psi}{\partial t_{n}}-\frac{\partial \Psi}{\partial \Phi}\frac{\partial \Phi}{\partial t_{n}}\right )d_{\lambda }\Phi \nonumber\\                              &=&res \Phi ^{i-1}\left (\left \{ H_{n},\Psi \right \}-\left \{ H_{n},\Phi \right \}\frac{\partial \Psi}{\partial \Phi} \right )d_{\lambda }\Phi \nonumber\\
        &=&res \Phi ^{i-1}\left [ \left ( \lambda H_{n,\lambda }\Psi _{x} -\lambda H_{n,x}\Psi _{\lambda }\right )-\left ( \lambda H_{n,\lambda }\Phi _{x} -\lambda H_{n,x}\Phi _{\lambda } \right )\frac{\partial \Psi}{\partial \Phi} \right ]d_{\lambda }\Phi \nonumber\\
        &=&res \Phi ^{i-1}\left [ H_{n,\lambda }\left ( \lambda \Psi _{x}\Phi _{\lambda } -\lambda \Phi _{x}\Psi _{\lambda }\right )-H_{n,x}\left ( \lambda \Psi _{\lambda }\Phi _{\lambda } -\lambda \Phi _{\lambda }\Psi _{\lambda } \right )\right ]d{\lambda }\nonumber\\
        &=&res \Phi ^{i-1}\left ( \frac{\partial H_{n} }{\partial \lambda }\Phi \right )d\lambda \nonumber\\
        &=&res \Phi ^{i}d_\lambda H_{n}.\nonumber
    \end{eqnarray}
\end{proof} 

Based on the above preparation, we will give an existence theorem for the tau function.
\begin{theorem}\label{the1}
    For the hierarchy of dispersionless equation \eqref{1}, there exists $\tau$ function satisfying
    \begin{eqnarray}
        d\log{\tau }=v_{n}d t_{n},\nonumber
    \end{eqnarray}
    in which ``d" represents the differentiation of $t_{n}$.
\end{theorem}
\begin{proof}
In fact, to prove the existence of $\tau$ function, we just need to show that the right-hand side of the equation is in closed form. It turn out to be
\begin{eqnarray}
    \frac{\partial v_{n} }{\partial t_{m} }= \frac{\partial v_{m} }{\partial t_{n} }.\nonumber
\end{eqnarray}
Then, according to the previous lemma, the following results are available, 
\begin{eqnarray}
    \frac{\partial v_{n} }{\partial t_{m} }- \frac{\partial v_{m} }{\partial t_{n} }
    &=&res \Phi ^{n}d_{\lambda }H_{m}-res \Phi ^{m}d_{\lambda }H_{n} \nonumber\\
    &=&res\left ( \left (\Phi ^{n} \right )_{\ge 0}+\left (\Phi ^{n} \right) _{< 0} \right)d_{\lambda }\left ( \Phi ^{m} \right ) _{\ge 0}-res\left ( \left ( \Phi ^{m} \right ) _{\ge 0}+\left ( \Phi ^{m} \right ) _{< 0} \right)d_{\lambda }\left ( \Phi ^{n} \right ) _{\ge 0} \nonumber\\
    &=&res\left (\Phi ^{n} \right) _{< 0}d_{\lambda }\left ( \Phi ^{m} \right ) _{\ge 0}-res\left (\Phi ^{m} \right) _{< 0}d_{\lambda }\left ( \Phi ^{n} \right ) _{\ge 0} \nonumber\\
    &=&res\left (\Phi ^{n} \right) _{< 0}d_{\lambda }\left ( \Phi ^{m} \right ) _{\ge 0}+res\left ( \Phi ^{n} \right ) _{\ge 0}d_{\lambda }\left (\Phi ^{m} \right) _{< 0} \nonumber\\
    &=&res\Phi ^{n}d_{\lambda }\Phi ^{m} \nonumber\\
    &=&res m\Phi ^{m+n-1}d_{\lambda }\Phi \nonumber\\
    &=&m\delta _{m+n,-1}. \nonumber
\end{eqnarray}    
Since $m,n$ are positive integers, the above equation is vanished. Eventually, the existence of the $\tau$ function is proved.
\end{proof}

\bigskip

\section{\sc \bf The twistor structure of the hierarchy}
In this section, we modify the eigenfunctions $\Phi$ and $\Psi$ using the dressing function. Then we further construct the twistor structure of the hierarchy.
\begin{proposition}
    Let $\Phi$ and $\Psi$ be solutions of the hierarchy, and there exists a dressing function $\varphi$ satisfying the following equations
    \begin{eqnarray}
        \Phi &=& e^{ad \; \varphi }\left ( \lambda \right ),\;\;\;
        \Psi = e^{ad \; \varphi }  \left(x+\sum_{n=1}^{\infty }\frac{t_{n} }{n}\lambda ^{n} \right ) ,\label{dressing} \\
        \bigtriangledown _{t_{n},\varphi }\varphi &=& -\left ( \frac{1}{n^{2} }e^{ad \; \varphi }(\lambda ^{n} ) \right )_{\le -1},\;\;\;n=1,2,\cdots, \label{tn} 
    \end{eqnarray}
    in which 
    \begin{eqnarray}
        \varphi (t)=\sum_{n=1}^{\infty}\varphi _{n}(t)\lambda ^{-n}.\nonumber
    \end{eqnarray}
\end{proposition}
     
Indeed, observing the above results, we derive an alternative expression for the eigenfunctions, one can apply such expression in the following proof of the  twistor structure of the hierarchy.
\begin{proposition}
    If $\varphi (t)=\sum_{n=1}^{\infty}\varphi _{n}(t)\lambda ^{-n}$ satisfies the equations of \eqref{tn}, then $\Phi$ and $\Psi$ defined by \eqref{dressing} are solutions of the hierarchy.
\end{proposition}
\begin{proof}
    For Laurent series $\Phi$ and $\Psi$, they follow from the definition, then the Lax equation represents
    \begin{eqnarray}
        \partial _{t_{n}}\Phi =\partial _{t_{n}} \left ( e^{ ad \; \varphi}(\lambda )\right ) &=&\left \{ \bigtriangledown _{t_{n},\varphi}\varphi, e^{ad \; \varphi}(\lambda ) \right \} \nonumber\\
        &=&e^{e^{ ad \; \varphi}}\left \{ e^{- ad \; \varphi}\bigtriangledown _{t_{n},\varphi }\varphi,\lambda  \right \} \nonumber\\
        &=&e^{ad \; \varphi}\left \{ e^{-ad \; \varphi}\left ( -\frac{1}{n^{2}}e^{ad \; \varphi}(\lambda ^{n}) \right )_{\le -1},\lambda \right \} \nonumber\\
        &=&e^{ad \; \varphi}\left \{ e^{-ad \; \varphi}H_{n},\lambda  \right \} \nonumber\\
        &=&\left \{ H_{n},\Phi \right \}.\nonumber
    \end{eqnarray}
    Another Lax equation reads as
    \begin{eqnarray}
        \partial _{t_{n}}(\Psi )&=&\partial _{t_{n}}\left (e^{ad \; \varphi} \left(x+\sum_{n=1}^{\infty }\frac{t_{n} }{n}\lambda ^{n} \right) \right) \nonumber\\
        &=&\left \{ \bigtriangledown _{t_{n},\varphi }\varphi , \Psi \right \}+e^{ad \; \varphi }\left ( \frac{1}{n}\lambda ^{n} \right ) \nonumber\\
        &=& \left \{H_{n}-\frac{1}{n^{2}}e^{ad \; \varphi}(\lambda ^{n}) , \Psi \right \}+ e^{ad \; \varphi }\left ( \frac{1}{n}\lambda ^{n} \right )  \nonumber\\     
        &=&e^{ad \; \varphi}\left ( \left \{ e^{-ad \; \varphi}H_{n}-\frac{1}{n^{2}}\lambda ^{n} , x+\sum_{n=1}^{\infty }\frac{t_{n} }{n}\lambda ^{n} \right \}+\frac{1}{n}\lambda ^{n} \right ) \nonumber\\
        &=&e^{ad \; \varphi}\left ( \left \{ e^{-ad \; \varphi}H_{n} , x+\sum_{n=1}^{\infty }\frac{t_{n} }{n}\lambda ^{n} \right \} \right ) \nonumber\\
        &=&\left \{ H_{n},\Psi \right \}. \nonumber
    \end{eqnarray}
\end{proof}
\begin{definition}
    Let there are a pair of two functions $\left ( f\left ( \lambda ,x \right ), g\left ( \lambda ,x \right ) \right )$ and a pair of Laurent series $\left ( \Phi, \Psi \right )$, where $\Phi$ and $\Psi$ are forms of equations \eqref{Phi} and \eqref{Psi} respectively. Suppose they satisfy the canonical Poisson relation $\left \{f\left ( \lambda ,x \right ), g\left ( \lambda ,x \right ) \right \}=f\left ( \lambda ,x \right )$, $F=f\left ( \Phi, \Psi \right )$ and $G=g\left ( \Phi, \Psi \right )$ are both Taylor in $\lambda$ as
    \begin{eqnarray}
        \left ( f\left ( \Phi ,\Psi \right ) \right )_{\le -1}=0,\;\;\;\left ( g\left ( \Phi ,\Psi \right ) \right )_{\le -1}=0.\nonumber
    \end{eqnarray} 
    Then the pair $\left ( \Phi, \Psi \right )$ is a solution of the hierarchy, for this reason, $\left ( f,g \right )$ is the twistor structure of this solution. 
 \end{definition}

\begin{theorem}
The hierarchy with twistor structure in their solutions, i.e., if $\left ( \Phi, \Psi \right )$ is a solution to the dispersionless equation \eqref{1}, then there exists twistor structure $\left ( f,g \right )$ of this hierarchy.    
\end{theorem}
\begin{proof}
    First, we show that $\left \{ f,g\right \}=f$. Let $exp\left ( ad \; \varphi( \Vec{t} ) \right )$ be the dressing operator corresponding to $\left ( \Phi, \Psi \right )$. When $\vec{t}=\vec{0}, \;\vec{t}=(t_{1}, t_{2}, \cdots)$. Set 
    \begin{eqnarray}
        f\left ( \lambda ,x \right )=e^{-ad \; \varphi }|_{\vec{t}=\vec{0}} \; \lambda ,\;\;\; g\left ( \lambda ,x \right )=e^{-ad \; \varphi }|_{\vec{t}=\vec{0}} \; x .\label{fg1}
    \end{eqnarray}
    In fact,
    \begin{eqnarray}
        \left \{ f\left ( \lambda ,x \right ),g\left ( \lambda ,x \right ) \right \} =e^{-ad \; \varphi }|_{\vec{t}=\vec{0}}\left \{ \lambda ,x \right \}=e^{-ad \; \varphi }|_{\vec{t}=\vec{0}} \lambda =f\left ( \lambda ,x \right ) .\nonumber
    \end{eqnarray}
    For equation \eqref{dressing},
    \begin{eqnarray}
        \Phi |_{\vec{t}=\vec{0}} =e^{ad \; \varphi } |_{\vec{t}=\vec{0}} \; \lambda ,\;\;\; \Psi |_{\vec{t}=\vec{0}}=e^{ad \; \varphi }|_{\vec{t}=\vec{0}} \; x. \nonumber
    \end{eqnarray}
    Then, the following relations are obtained from equation \eqref{fg1},
    \begin{eqnarray}
        \begin{aligned}\label{fg2}
            f\left (\Phi |_{\vec{t}=\vec{0}},\Psi |_{\vec{t}=\vec{0}} \right)
            &&=e^{ad \; \varphi } |_{\vec{t}=\vec{0}} f\left ( \lambda ,x \right )=\lambda,\\
            g\left (\Phi |_{\vec{t}=\vec{0}},\Psi |_{\vec{t}=\vec{0}} \right)
            &&=e^{ad \; \varphi } |_{\vec{t}=\vec{0}} g\left ( \lambda ,x \right )=x. 
        \end{aligned}
    \end{eqnarray}
    As a result,
    \begin{eqnarray}
        \left \{ f \left ( \Phi,\Psi \right ),g\left ( \Phi,\Psi \right ) \right \}=e^{ad \; \varphi }|_{\vec{t}=\vec{0}} \left \{ f\left ( \lambda ,x \right ),g\left ( \lambda ,x \right ) \right \}=e^{ad \; \varphi }|_{\vec{t}=\vec{0}} f\left ( \lambda ,x \right )=f \left ( \Phi,\Psi \right ). \nonumber
    \end{eqnarray}
    
    Secondly, prove that $\left ( f\left ( \Phi ,\Psi \right ) \right )_{\le -1}=0$. Since $\Phi$ and $\Psi$ satisfy the Lax equations \eqref{hierarchy}, then
    \begin{eqnarray}
        \frac{\partial f\left ( \Phi ,\Psi  \right ) }{\partial t_{n}}
        &=&f_{\Phi }\left ( \Phi ,\Psi \right )\cdot \Phi _{t_{n}}+f_{\Psi }\left ( \Phi ,\Psi \right )\cdot \Psi _{t_{n}} \nonumber\\
        &=&f_{\Phi }\left ( \Phi ,\Psi \right )\left \{ H_{n},\Phi \right \}+f_{\Psi }\left ( \Phi ,\Psi \right )\left \{ H_{n},\Psi \right \} \nonumber\\
        &=&\left \{ H_{n},f\left ( \Phi ,\Psi  \right ) \right \}.\label{ftn}
    \end{eqnarray}
    When $\vec{t}=\vec{0}$, based on the equation \eqref{fg2} and \eqref{ftn}, the result is as follows
    \begin{eqnarray}
        \frac{\partial f\left ( \Phi ,\Psi  \right ) }{\partial t_{n}}|_{\vec{t}=\vec{0}}
        &=&\left \{ H_{n},f\left ( \Phi=0 ,\Psi=0  \right ) \right \}\nonumber\\
        &=&\left \{ H_{n},\lambda \right \} \nonumber\\
        &=&-\lambda H_{n,x}. \nonumber
    \end{eqnarray}
    Obviously, the above equation does not contain the negative part of $\lambda$. Beyond that, for equation $\left (\partial /\partial _{t} \right)^{\alpha}f\left (\Phi,\Psi \right)|_{(\vec{t}=\vec{0} )}$, regardless of the value of $\alpha$, the Taylor expansion coefficients of this equation at $\vec{t}=\vec{0}$ do not contain negative terms of $\lambda$. Ultimately we can prove that $\left ( f\left ( \Phi ,\Psi \right ) \right )_{\le -1}=0$. The same reasoning leads to $\left ( g\left ( \Phi ,\Psi \right ) \right )_{\le -1}=0$.
\end{proof}

\bigskip

\section{\sc \bf The relevant nonlinear Riemann-Hilbert problem for constructing solutions}
In this section, we investigate the long-time behaviour and the possible wave breaking properties which are important aspects in the study of dispersionless equations. First and the key step is to relate the dispersionless equations to the nonlinear Riemann-Hilbert problem using the Manakov-Santini method. In fact, solving equation \eqref{1} can be transformed into studying the nonlinear Riemann-Hilbert problem. This result is shown in the following theorem.
\begin{theorem}
    Consider the vector nonlinear Riemann-Hilbert problem on the real line as
    \begin{eqnarray}
        \vec{\pi }^{+}(\lambda)=\vec{R} ( \vec{\pi }^{-}(\lambda)),\;\;\;\lambda \in \mathbb{R} \label{RH1},
    \end{eqnarray}
     in which $\vec{\pi }^{+}(\lambda ) ,\; \vec{\pi }^{-}(\lambda ) \in \mathbb{C}^{2}$ are two-dimensional vector functions resolved in the upper and lower halves of the complex $\lambda$ plane, respectively,  and normalised them into the following form as
     \begin{eqnarray}                                                           \vec{\pi }^{\pm}(\lambda )=\begin{pmatrix}
         \pi _{1}^{\pm }(\lambda) \\[10pt]
         \pi _{2}^{\pm }(\lambda) 
         \end{pmatrix}
         = \begin{pmatrix}
           \lambda +2u\\[10pt]
           \frac{t}{2}\lambda ^{2}+ut\lambda +y\lambda +x+2uy+2u^{2}t+4t\partial _{x}^{-1}u_{y}
           \end{pmatrix} 
           + o (\lambda ^{-1} ) ,\;\;\;  \left | \lambda  \right |\gg 1. \label{pi}
     \end{eqnarray}
     In fact, 
     \begin{eqnarray}
         \partial _{x}^{-1}u=\lim_{\lambda  \to \infty}\frac{\lambda }{2}\left (\pi _{2}^{\pm }(\lambda ) -y\pi _{1}^{\pm }(\lambda)-\frac{t}{2}\pi_{1}^{\pm ^2}(\lambda) -x \right ),
     \end{eqnarray}
     and the differentiable spectral datum $\vec{R}(\vec{\xi})=\left ( R_{1}(\xi _{1},\xi _{2}) ,R_{2}(\xi _{1},\xi _{2}) \right )$, the vector $\vec{\xi }\in \mathbb{C}^{2},\;\lambda\in\mathbb{R}$, satisfying the following given realistic constraint as
     \begin{eqnarray}
         \vec{R}(\overline{\vec{R}(\bar{\vec{\xi }})})=\vec{\xi }. \label{constraint}
     \end{eqnarray}
     
     Applying the vector fileds Lax pair $L_{1}$, $L_{2}$ defined by equations \eqref{L1} and \eqref{L2} on the above nonlinear Riemann-Hilbert problem \eqref{RH1}, there are $L_{j}\vec{\pi }^{+}(\lambda)=JL_{j}\vec{\pi }^{-}(\lambda)$, $j=1,2$, where $J$ satisfies the Jacobian matrix $J_{mn} =\partial R_{m}/\partial \xi _{n}$, $(m,n=1,2)$. Then, assuming that both forms of the solution are unique, $\vec{\pi }^{\pm}(\lambda)$ of the Riemann-Hilbert problem \eqref{RH1} are common eigenfunctions of the vector fields, i.e., $L_{j}\vec{\pi }^{\pm}(\lambda)=0$ and $u$ is the solution to the equation \eqref{1}. Combining $\vec{R}(\vec{\xi})$ with constraint \eqref{constraint}, then the solutions $u$ is real, i.e., $u\in\mathbb{R}$.   
\end{theorem}
\begin{proof}
Let the vector fileds Lax pair $L_{1}$ and $L_{2}$ act on $\vec{\pi }^{\pm}(\lambda )$ in equation \eqref{pi}. In fact, when $\lambda \to \infty$, as a result,
\begin{eqnarray}
    L_{j}\vec{\pi }^{\pm}(\lambda)\to 0,\;\;\; j=1,2.\nonumber
\end{eqnarray}
Based on the nonlinear Riemann-Hilbert problem linearized version
\begin{eqnarray}
    L_{j}\vec{\pi }^{+}(\lambda)=JL_{j}\vec{\pi }^{-}(\lambda),\;\;\;j=1,2 ,\nonumber
\end{eqnarray}
where $J$ is the Jacobi matrix satisfying 
\begin{eqnarray}
    J=\begin{pmatrix}
      \dfrac {\partial R_{1}}{\partial \xi _{1} } & \dfrac {\partial R_{2}}{\partial \xi _{1} } \\[10pt]
      \dfrac {\partial R_{1}}{\partial \xi _{2} } & \dfrac {\partial R_{2}}{\partial \xi _{2} }
      \end{pmatrix}. \nonumber
\end{eqnarray}
Obviously, combining equation \eqref{RH1}, the vectors $L_{j}\vec{\pi }^{\pm}(\lambda)$ solve the linearised Riemann-Hilbert problem.
By uniqueness, we infer that $\vec{\pi }^{\pm}(\lambda)$ are shared eigenfunctions of the Lax pair $L_{j}$, obtaining 
 \begin{eqnarray}
     L_{j}\vec{\pi }^{\pm}(\lambda)=0,\;\;\;j=1,2. \nonumber
 \end{eqnarray}
This indicates that the $\vec{\pi }^{\pm}(\lambda)$ are the solutions of the nonlinear Riemann-Hilbert problem. From the equation \eqref{pi}, one has
\begin{eqnarray}
     \partial _{x}^{-1}u=\lim_{\lambda  \to \infty}\frac{\lambda }{2}\left (\pi _{2}^{\pm }(\lambda ) -y\pi _{1}^{\pm }(\lambda)-\frac{t}{2}\pi_{1}^{\pm ^2}(\lambda) -x \right ). \nonumber
\end{eqnarray}
Actually, it is the solution to dispersionless system \eqref{1}. 

Similarly, by uniqueness, combining constraint \eqref{constraint} of $\vec{R}(\vec{\xi})$ and the nonlinear Riemann-Hilbert problem \eqref{RH1}, we can get the following expression 
\begin{eqnarray}
    \overline{\vec{\pi}^{+}}(\lambda) =\vec{\pi }^{-}(\lambda).\nonumber
\end{eqnarray}
It follows that, the solution $u$ is real, i.e., $u\in\mathbb{R}$.
\end{proof}

\bigskip
\bigskip
\textbf{Acknowledgements:} This work is supported by the National Natural Science Foundation of China under Grant Nos. 12271136, 12171133 and 12171132.

\end{document}